\documentclass[review]{siamonline220329}

\nolinenumbers

\usepackage{lipsum}
\usepackage{amsfonts}
\usepackage{graphicx}
\usepackage{epstopdf}
\usepackage{algorithmic}
\ifpdf
  \DeclareGraphicsExtensions{.eps,.pdf,.png,.jpg}
\else
  \DeclareGraphicsExtensions{.eps}
\fi

\DeclareMathOperator{\ww}{w}
\DeclareMathOperator{\sgn}{sgn}

\usepackage{stmaryrd}

\newcommand{\inhib}{\relbar\mapsfromchar}

\numberwithin{theorem}{section}

\newtheorem{remark}[theorem]{Remark}        
\newtheorem{example}[theorem]{\emph{Example}}        

\newcommand{\TheTitle}{Algebraic network reconstruction of discrete dynamical systems} 
\newcommand{\TheAuthors}{H. A. Harrington, M. Stillman, and A.~Veliz-Cuba}

\headers{\TheTitle}{\TheAuthors}

\title{{\TheTitle}\thanks{
\funding{H.A.H. gratefully acknowledges funding from EPSRC EP/R018472/1, EP/R005125/1 and EP/T001968/1, a Royal Society University Research Fellowship RGF$\backslash$EA$\backslash$201074 and UF150238. For the purpose of Open Access, the authors have applied a CC BY public copyright licence to any  Author Accepted Manuscript (AAM) version arising from this submission.   A.VC. was partially supported by the Simons Foundation grant 516088.}}}

\author{
  Heather A. Harrington\thanks{Mathematical Institute, University of Oxford
    (\email{harrington@maths.ox.ac.uk}),}
  \and
  Mike Stillman\thanks{Department of Mathematics, Cornell University (\email{mes15@cornell.edu}).}
  \and
  Alan Veliz-Cuba \thanks{Department of Mathematics, University of Dayton (\email{avelizcuba1@udayton.edu})
}}

\usepackage{amsopn}

\usepackage{fancybox}
\usepackage{todonotes}
\usepackage{menukeys}
\ifpdf
\hypersetup{
  pdftitle={\TheTitle},
  pdfauthor={\TheAuthors}
}
\fi

\begin{document}

\maketitle

\begin{abstract}
We present a computational algebra solution to reverse engineering the network structure of discrete dynamical systems from data. We use monomial ideals to determine dependencies between variables that encode constraints on the possible wiring diagrams underlying the process generating the discrete-time, continuous-space data. Our work assumes that each variable is either monotone increasing or decreasing. We prove that with enough data, even in the presence of small noise, our method can reconstruct the correct unique wiring diagram.

\end{abstract}


%

\section{Introduction}

Many biological systems have been modeled using discrete-time systems of the form $f=(f_1,\ldots,f_n):X^n\rightarrow X^n$. Here, each coordinate function $f_i$ describes how the behavior of variable $i$ depends on the other variables. Such a modeling framework has been used successfully to study biological features such as equilibrium and periodic behavior \cite{Eager2016,Townley2012,Arat2015,Velizlacop}. 

In the cases where $f$ is unknown, one must infer the structure of the network from data. We refer to this inverse problem as the network reconstruction problem, which has been studied in the case that $X$ is a finite set. In this setting, tools from computational algebra were used to find the best networks given data \cite{Veliz-Cuba2012,Jarrah2012}; experimental data are typically continuous so the data must be quantized with discretizations algorithms \cite{Dimitrova2010}. Theory and practical implications, such as experimental design for network reconstruction, has recently been studied \cite{dimitrova2022algebraic}. Recent work has proposed how to reconstruct Boolean functions \cite{sun2022data}. 
 In this manuscript we study the problem of reconstructing the network structure of discrete-time continuous-space dynamical systems. By focusing on continuous-space dynamical systems, we eliminate the need of discretization algorithms and their unknown effect on network reconstruction. Furthermore, experimental data are noisy due to measurement errors and stochasticity, so we also study the effect of noisy data on network reconstruction.

For simplicity in the presentation we consider dynamical systems defined by 

\[ f=(f_1,\ldots,f_n):[0,1]^n\rightarrow [0,1]^n, \]
but we remark that our results are valid even if the space is not bounded.
 The dynamics of such systems are given by iteration of $f$, $x(t+1)=f(x(t))$. The coordinate functions $f_1,\ldots f_n$ describe how a variable depends on the others and determine the structure of the network. The structure of the network is given by a signed directed graph with nodes $x_1,\ldots,x_n$ (or $1,\ldots,n$) such that there is an edge from $x_i$ to $x_j$ if $f_j$ depends on $x_i$. The sign of this edge is positive if $f_j$ is increasing with respect to $x_i$, and is negative if $f_j$ is decreasing with respect to $x_i$. We will focus on dynamical systems where each edge has a sign; that is, dynamical systems where each $f_j$ is either monotone increasing or decreasing with respect to its variables. We call such functions monotone.

Our goal is to use dynamical information of $f:[0,1]^n\rightarrow [0,1]^n$ to reconstruct the structure of the network. Namely, we want to determine which edges appear and their signs. The statement of the problem is as follows. Consider $P$ a finite subset of $[0,1]^n$ such that $f|_P$ is known. How can we use knowledge of $f|_P$ to reconstruct the network?

The paper is organized as follows. We present notation and definitions in \cref{sec:pre}. Algebraic notions  for discrete dynamical systems are presented in \cref{sec:alg}.
Our main results are in
\cref{sec:main} with selection of diagrams in \cref{sec:score}, and the conclusions follow in
\cref{sec:conclusions}.

\section{Preliminaries}
\label{sec:pre}

We first consider the case of determining which variables appear in a
single coordinate function. Namely, consider a (possibly unknown)
function $h:[0,1]^n\rightarrow [0,1]$ such that $h|_P$ is known, where
$P$ is a finite set of points in $[0,1]^n$. We will refer to the pair $D = (P,
h|_P)$ as \emph{observed data}.  We will focus on determining from
observed data which variables affect $h$ as well as whether $h$ is
increasing or decreasing with respect to these variables.

We introduce terminology required to study this situation.

\begin{definition}\rm
Suppose that $h:[0,1]^n\rightarrow [0,1]$ is a function. We say that $h$ is

(1) \emph{independent of $x_i$} if for all choices
$\{c_1, \ldots, c_{i-1}, c_{i+1}, \ldots, c_n\}$, the function of one variable
  $g(x_i) := h(c_1, \ldots, c_{i-1}, x_i, c_{i+1}, \ldots, c_n)$ is a constant function in $x_i$.

(2)
\emph{monotone increasing in the variable $x_i$} if
all the $g(x_i)$ as in (1) are monotone increasing functions (i.e. for
$a < b$, $g(a) \le g(b)$).

(3)
 \emph{monotone decreasing in the variable $x_i$} if
all the $g(x_i)$ as in (1) are monotone decreasing functions (i.e. for
$a < b$, $g(a) \ge g(b)$).

(4)  \emph{monotone} if for each $x_i$, it satisfies one of (1), (2), or (3).
  \end{definition}

\begin{remark}
  In this paper, all monotone functions which appear will be continuous.
\end{remark}

Given a discrete dynamical system $f = (f_1, \ldots, f_n)$, where each $f_i$ is monotone, then the monotonicity information as defined in Definition 2.1 determines the network structure, i.e., a signed directed graph. The incoming edges to node $x_i$ are determined by the monotonicity information for $f_i$. We call this information for $f_i$ the local wiring diagram of variable $x_i$.

\begin{definition}\rm
  (a) A \emph{local wiring diagram} is a set $w$ with elements of
  the form $(x_i,s)$, where $s\in\{1,-1\}$ and any variable $x_i$ appears
  at most once.

  (b) Given a monotone function
  $h:[0,1]^n\rightarrow [0,1]$, the \emph{local wiring diagram of $h$} is
  \begin{align*} {\ww}(h)  := & \left\{ (x_k, 1) \mid \text{$h$ is monotone increasing in $x_k$, but not independent of$ x_k$} \right\} \\
    & \cup
    \left\{ (x_k, -1) \mid \text{$h$ is monotone decreasing in $x_k$, but not independent of $x_k$} \right\}.
    \end{align*}
\end{definition}

\begin{example} Consider $h:[0,1]^5\rightarrow [0,1]$ defined by $h(x)=\frac{x_1}{1+x_1}\frac{x_3^2}{1+x_3^2} \frac{1}{1+x_5}$. Since $h$ is increasing with respect to $x_1$ and $x_3$, and decreasing with respect to $x_5$, it follows that the local wiring diagram of $h$ is $\ww(h)=\{(x_1,1),(x_3,1),(x_5,-1)\}$.
\end{example}

\begin{example} Consider $h:[0,1]^5\rightarrow [0,1]$ defined by $h(x)=1$. Since $h$ is independent of all variables, its local wiring diagram is the empty set, $\ww(h)=\{ \ \}$.
\end{example}

\begin{definition}\label{def:minimal}\rm
  Given observed data $D = (P, h|_P)$, we say that a local wiring diagram
  \[w=\{(x_{k_1},s_1), (x_{k_2},s_2), \ldots,(x_{k_m},s_m)\}\] is
  \emph{consistent with the data $D$} if 
  there exists a continuous monotone function $h^*:[0,1]^n\rightarrow [0,1]$ such
  that $\left.h^*\right|_P=h|_P$ and $\ww(h^*)\subseteq w$.  The local
  wiring diagram $w$ is called a \emph{minimal local wiring diagram} if it
  is consistent with the data and does not contain a smaller (with
  respect to inclusion) consistent local wiring diagram.
\end{definition}

We denote by $W_D$ the set of all local wiring diagrams consistent
with the data $D$.

The two conditions above mean that a minimal local wiring diagram is a set of
variables (with signs) that is consistent with the data (condition 1) and is minimal
with respect to inclusion (condition 2). This definition incorporates
the biological perspective that local wiring diagrams should be as simple as
possible while still being consistent with the data.

\begin{example}\label{running:example}\rm
Consider $h:[0,1]^3 \rightarrow [0,1]$ monotone and suppose that
$h(.1,.8,.3)=.2$, $h(.9,.5,.1)=.5$, and $h(.5,.3,.9)=.7$. In this case
$P=\{(.1,.8,.3),(.9,.5,.1),(.5,.3,.9)\}$ and $h|_P$ is known. We claim
that there are two minimal local wiring diagrams, namely $w_1=\{(x_2,-1)\}$
and $w_2=\{(x_1,1),(x_3,1)\}$. We remark that at this point we are
using this example to illustrate the definition only, not to show how
minimal wiring diagrams are found.

The local wiring diagram $w_1$ is a minimal local wiring diagram because there exists a monotone
function $h^*:[0,1]^3\rightarrow [0,1]$ given by $h^*(x)=1-x_2$ such
that $h^*|_P=h|_P$. That is, $w_1$ satisfies the first condition in
the definition. Now, suppose $g:[0,1]^n\rightarrow [0,1]$ is another
monotone function such that $\left.g\right|_P=h|_P$ and
$\ww(g)\subseteq w_1$. We then have two cases
$\ww(g)=\{\ \}$ or $\ww(g)=\{(x_2,-1)\}$. Since
$g|_P=h|_P$, we see that $g$ cannot be constant, so
$\ww(g)=\{(x_2,-1)\}$. Thus, $w_1$ satisfies the second
condition in the definition.

To show that $w_2$ is also a minimal wiring diagram, we first observe
that $h^*:[0,1]^3\rightarrow [0,1]$ given by $h^*(x)=\frac{x_1+x_3}{2}$
is a monotone function that satisfies $h^*|_P=h|_P$ and
$\ww(h^*)=w_2$. Thus, $w_2$ satisfies the first condition in the
definition. Second, suppose $g:[0,1]^n\rightarrow [0,1]$ is another
monotone function such that $\left.g\right|_P=h|_P$ and
$\ww(g)\subseteq w_2$. As shown with $w_1$, $g$ cannot be
constant, so we have to show that $\ww(g)\neq \{(x_1,1)\}$ and
$\ww(g)\neq \{(x_3,1)\}$. Since $g(.9,.5,.1) =
h(.9,.5,.1)=.5$, $g(.5,.3,.9)=h(.5,.3,.9)=.7$, and $g$ is increasing
with respect to $x_1$, $g$ cannot depend only on $x_1$. Similarly, $g$
cannot depend on $x_3$ only. Then, $\ww(g)=w_2$ and so $w_2$
satisfies the second condition in the definition.

Note that the functions $h^*$ used to show that $w_1$ and $w_2$ are
minimal are not necessarily unique. For example, consider
the functions $h^*(x)=b+\frac{a}{c^n+x_2^n}$ for $w_1$ and
$h^*(x)=b+a\frac{x_1^n}{c^n+x_1^n}\frac{x_3^n}{d^n+x_3^n}$ for $w_2$
(with appropriate values for $a,b,c,d,n$).
\end{example}

\begin{remark}
The precise form of the functions $h^*$ that we use in the definition
is not known. Therefore, finding the minimal local wiring diagrams by
constructing functions that are consistent with the data is not
feasible.
\end{remark}

To make network reconstruction feasible we will show that we can find minimal local wiring diagrams without having to construct the functions. First we need to define what it means that $w$ is consistent with the data using the data only. Namely, if $p,p'\in P$
and $h(p)<h(p')$, then the increase in the output has to correspond to
an increase in an activator or a decrease in a repressor; that is, an
increase in some $i$-th entry of the input such that $(x_i,1)\in w$,
or to a decrease in some $i$-th entry of the input such that
$(x_i,-1)\in w$. The following definition formalizes this idea.

\begin{definition}\rm
Let $p,p'\in P$ such that $h(p)<h(p')$. We say that the local wiring diagram $w$ is \emph{consistent}
with the pair $(p,p')\in P^2$ if for some $i$, $p_i<p'_i$ and $(x_i,1)\in w$, or
$p_i>p'_i$ and $(x_i,-1)\in w$. Equivalently,
for some $i$, $(x_i,\sgn(p'_i-p_i))\in w$.
We denote by $W_{(p,p')}$ the set of all local wiring
diagrams that are consistent with the pair $(p,p')$.
\end{definition}

Note that if either $W_D$ or $W_{(p,p')}$ contains a local wiring diagram $w$, then it contains
every local wiring diagram $w'$ for which $w \subseteq w'$.  Therefore, in order
to describe these sets, we need only consider minimal elements (with respect to inclusion).
If $w_1, \ldots, w_r$ are the minimal elements of $W_D$, we often write
\[ W_D = \langle w_1, w_2, \ldots, w_r \rangle. \]

\begin{example}\label{running2}\rm
  Continuing $\ref{running:example}$, 
  let $P = \{ p_1, p_2, p_3\}$, where $p_1 = (.1,.8,.3)$, $p_2 =
  (.9,.5,.1)$, and $p_3 = (.5,.3,.9)$.
  Consider again $h:[0,1]^3 \rightarrow [0,1]$ monotone and suppose that
  $h(p_1) = .2$, $h(p_2) = .5$, and $h(p_3) = .7$.

  First let us find the elements of $W_{(p_1,p_3)}$
  (note $h(p_1)<h(p_3)$). The increase in the output
  has to correspond to the increase in $x_1$, the decrease in $x_2$,
  or the increase in $x_3$. It follows that $w\in W_{(p_1,p_3)}$ if and
  only if $(x_1,1)\in w$ or $(x_2,-1)\in w$ or $(x_3,1)\in w$.

  Similarly, $w \in W_{(p_1,p_2)}$ if and only if $(x_1,1)\in w$ or $(x_2,-1)\in w$ or $(x_3,-1)\in w$.
  Also, $w \in W_{(p_2,p_3)}$ if and only if $(x_1,-1)\in w$ or $(x_2,-1)\in w$ or $(x_3,1)\in w$.
\end{example}

The following Theorem is a continuous-space version of the discrete-space Lemma 2.4-Theorem 2.5 in \cite{Veliz-Cuba2012} and requires a different proof.

\begin{theorem}\label{thm:def}
  Let $D = (P, h|_P)$ be observed data. Then the set of local wiring diagrams
  consistent with the data $D$ is exactly the set of wiring diagrams which are
  consistent with each pair of points $p,p' \in P$ satisfying $h(p) < h(p')$, that is:
\[ W_D=\bigcap_{\substack{(p,p')\in P^2 \\ h(p)<h(p')}} W_{(p,p')}. \]
\end{theorem}

\begin{proof}
  The left hand side is easily seen to be contained in the right hand side.
  For the opposite direction, consider $w$ in the right hand side (i.e. $w \in W_{{p,p'}}$ for all pairs of points of the data $D$
  with $h(p) < h(p')$) and without loss of generality assume $w=\{(x_1,1),\ldots,(x_k,1)\}$. Then, define $Q=\{(p_1,\ldots,p_k):p\in P\}$.

First, we claim that the data $D'=(Q,h|_P)$  is monotone increasing. That is, for $q,q'\in Q$ and  $v=h(p)$ and $v'=h(p')$,  if $q\leq q'$ (entrywise) then $v\leq v'$ (note that we are not saying $h$ is monotone). Indeed, by contradiction suppose $v'<v$, then since $w$ is consistent with $(p',p)$, there is $i$ such that $(x_i,1)\in w$ and $p'_i<p_i$. Then $1\leq i \leq k$ and $q'_i<q_i$. This contradicts the fact that  $q\leq q'$ entrywise. Thus, the data are monotone.

Second, we extend the data to cover a rectangular grid of values. Namely, for $y\in [0,1]^k$, we define 
$
g(y):=\max\{ v : y\leq q \text{ and } (q,v)\in D'\}.
$
 We remark that $g$ is monotone increasing, so the data we obtain by restricting $g$ to a rectangular grid will also be monotone increasing.

Third, since we have monotone data on a rectangular grid, we can use multilinear interpolation to obtain a continuous function $L:[0,1]^k\rightarrow [0,1]$ that fits the data on a grid. Then, if we define $h^*:[0,1]^n\rightarrow [0,1]$ by $h(x)=L(x_1,\ldots,x_k)$, it follows that $h^*|_P=h|_P$ and $\ww(h^*)\subseteq w$. This completes the proof.

  \end{proof}

\begin{remark} In the book \cite[Chapter 8]{schumaker}, there is a discussion on finding monotone spline functions that are differentiable in the two-dimensional case; these methods likely carry over to the n-dimensional case. 
\end{remark}

\begin{example}\rm
  Continuing with Example \ref{running:example}, \ref{running2}, let
  us find $W_D$. From this proposition, we want to find those local
  wiring diagrams which are in all three sets $W_{(p_i,p_j)}$.  One
  local wiring diagram that is in all three sets and therefore in
  $W_D$ is $w_1 = \{(x_2, -1)\}$.  Therefore any local wiring diagram
  which contains $(x_2, -1)$ is also in $W_D$.

  To find other local wiring diagrams $w$ in $W_D$, we may assume that $(x_2, -1)$ is
  not in $w$.  Since $w \in W_{(p_1,p_3)}$, either $(x_1,1) \in w$ or $(x_3,1) \in w$.
  So first suppose $(x_1,1) \in w$.  In this case, $w \in W_{(p_1,p_2)}$ only if $(x_3,1) \in w$.
  Since this will also imply that $w \in W_{(p_2,p_3)}$, we see that $w_2 = \{(x_1,1), (x_3,1)\}$
  is also in $W_D$.

  Finally, if we assume that $(x_2,-1) \not\in w$, and $(x_1,1) \not\in w$, then we would need
  $(x_3,-1) \in w$, but also we would need $(x_3, 1) \in w$, which cannot happen.  So there are
  no further minimal local wiring diagrams consistent with $D$, and therefore
  $W_D = \langle w_1, w_2 \rangle$.  

  In fact, $W_D$ consists of all wiring diagrams containing either $w_1$ or $w_2$.
  Therefore the minimal local wiring diagrams consistent with $D$ are $w_1$ and $w_2$.
  We often write $W_D = \langle w_1, w_2 \rangle$ to mean the set of all wiring diagrams
  containing either of these two elements.  There are 11 elements in $W_D$.
\end{example}

\begin{example}\rm
  Consider $h:[0,1]^2 \rightarrow [0,1]$ monotone and suppose that
  $h(.1,.1)=.2$, $h(.5,.3)=.4$, $h(.7,.2)=.6$ and $h(.8,.5)=.7$. It
  can be shown by inspection that $W_D=\{\{(x_1,1)\},
  \{(x_1,1),(x_2,1)\}, \{(x_1,1),(x_2,-1)\}\}$. Then, we see that
  there is a unique minimal element, $\{(x_1,1)\}$. By Theorem~\ref{thm:def},
  $\{(x_1,1)\}$ is the only minimal wiring diagram.
\end{example}

Theorem~\ref{thm:def} allows us to use numerical data and explore the
wiring diagram space without the need to explore the space of all
functions that could fit the data. In a sense, the theorem allows us
to work ``at the wiring diagram level''. 

\section{Algebraic approach for network reconstruction}\label{sec:alg}

This section extends the results for finite dynamical systems in
\cite{Veliz-Cuba2012} to discrete dynamical systems.  Let $D = (P, h|_P)$
be observed data.  Recall that $W_D$ is the set of local wiring diagrams
consistent with the data $D$.  In this section, we encode
$W_D$ algebraically with the help of Theorem~\ref{thm:def}.  The
problem of finding the minimal local wiring diagrams will be
transformed into a well known problem in computational algebra.

We now define three ideals in a polynomial ring $R$ in the variables $x_1, \ldots, x_n$.
The following definition encodes $W_{(p,p')}$ as an ideal of polynomials.
The intuition behind this definition is that for a local wiring diagram $W$
to be consistent with $(p,p'), h(p) < h(p')$, $W$ has to contain
$(x_i,\sgn(p'_i-p_i))$ for some $i$. This is formalized in the
following proposition.

\begin{definition}\rm
  Let $w$ be a local wiring diagram. Define the ideal
  \[
  \mathcal{I}_w:=\left\langle \{
  x_i-s_i| (x_i,s_i)\in w \} \right\rangle.
  \]

  This is a prime ideal generated by several linear polynomials.
\end{definition}

\begin{definition}\rm
Let $(p,p')\in P^2$ such that $h(p)<h(p')$. Define the ideal
\[
\mathcal{I}_{(p,p')} := \left\langle \prod_{p_i\neq p_i'} (x_i-\sgn(p'_i-p_i)) \right\rangle.
\]
This is an ideal generated by a single polynomial which is a product of linear polynomials.
\end{definition}

\begin{definition}\rm
  Let $D = (P, h|_P)$ be observed data.  Define the ideal
  \[\mathcal{I}_D:=\sum_{\substack{(p,p')\in P^2 \\ h(p)<h(p')}}  \mathcal{I}_{(p,p')}.
  \]

  This is an ideal generated by a number of nonlinear polynomials, each is a product of
  linear polynomials.
\end{definition}

\begin{proposition}\label{prop:wd_ideal}
A local wiring diagram $w$ is consistent with the pair $(p,p')$ (where $h(p) < h(p')$)
if and only if $\mathcal{I}_{(p,p')}\subseteq \mathcal{I}_w$. Furthermore, $w\in W_D$ if and only if $\mathcal{I}_D\subseteq \mathcal{I}_w$
\end{proposition}
\begin{proof}
First, suppose $w$ is consistent with $(p,p')$. Then, there is $j$
such that $(x_j,\sgn(p'_j-p_j))\in w$. Then
$x_j-\sgn(p'_j-p_j)$ is one of the generators of
$\mathcal{I}_w$ and $\langle x_j-\sgn(p'_j-p_j)\rangle
\subseteq \mathcal{I}_w$. Since $x_j-\sgn(p'_j-p_j)$ is a
factor of the generator of $\mathcal{I}_{(p,p')}$,
$\mathcal{I}_{(p,p')} \subseteq \langle
x_j-\sgn(p'_j-p_j)\rangle$. Thus
$\mathcal{I}_{(p,p')}\subseteq \mathcal{I}_w$.

Now, suppose $\mathcal{I}_{(p,p')}\subseteq \mathcal{I}_w$. Since
$\prod_{p_i\neq p_i'} (x_i-\sgn(p'_i-p_i))$ is in $\mathcal{I}_w$,
which is a prime ideal, one of the factors of the polynomial must be
in $\mathcal{I}_w$. Then, $x_j-\sgn(p'_j-p_j)\in
\mathcal{I}_w$ for some $j$ and so
$(x_j,\sgn(p'_j-p_j))\in w$. Thus, $w$ is consistent with
$(p,p')$.

To prove the second part of the proposition, note that $w\in W_D$ if and only if $w\in W_{(p,p')}$ for all pairs $(p,p')$ such that $h(p)<h(p')$, if and only if $\mathcal{I}_{(p,p')}\subseteq \mathcal{I}_w$ for all pairs $(p,p')$ such that $h(p)<h(p')$, if and only if $\mathcal{I}_D\subseteq \mathcal{I}_w$.
\end{proof}

The following proposition show us how to encode algebraically all
wiring diagrams that are consistent with $D$ and how to find the
minimal wiring diagrams algebraically.

\begin{proposition}\label{prop:minimal}
  Let $D = (P, h|_P)$ be observed data and consider $w$ to be a local wiring diagram. Then, 
 $w$ is a minimal local wiring diagram of $W_D$ if and only if $\mathcal{I}_w$ is a minimal
prime of $\mathcal{I}_D$.
\end{proposition}
\begin{proof}
First, note that since $\mathcal{I}_D$ is generated by products of $(x_i\pm 1)$, then
its minimal primes are of the form $\mathcal{I}_{w'}$ for some local wiring diagram $w'$, which by Proposition \ref{prop:wd_ideal} must be in $W_D$. We now proceed with the proof.

Suppose $w\in W_D$ is a minimal local wiring diagram. Since $\mathcal{I}_D  \subseteq \mathcal{I}_w$, then $\mathcal{I}_w$ must contain one of the minimal primes of $\mathcal{I}_D$, which will be of the form $\mathcal{I}_{w'}$ for some local wiring diagram $w'\in W_D$. Since $\mathcal{I}_{w'}\subseteq \mathcal{I}_{w}$, it follows that $w'\subseteq w$. We assumed $w$ is minimal, so $w'=w$ and $\mathcal{I}_{w'}=\mathcal{I}_w$. Thus, $\mathcal{I}_w$ is a minimal prime of $\mathcal{I}_D$.

Now consider a minimal prime of $\mathcal{I}_D$,  $\mathcal{I}_w$. Then, by Proposition \ref{prop:wd_ideal} $w\in W_D$. Denote with $w'$ the minimal local wiring diagram that is contained in $w$. Then, $\mathcal{I}_D\subseteq \mathcal{I}_{w'}\subseteq \mathcal{I}_w$, and since $\mathcal{I}_w$ is minimal we obtain $\mathcal{I}_{w'}= \mathcal{I}_w$. Thus $w'=w$ and $w$ is a minimal local wiring diagram.
\end{proof}

\begin{theorem}\label{thm:prob1}
Consider a monotone function $h:[0,1]^n\rightarrow [0,1]$ and suppose
we obtain data, $D$ by sampling points in $[0,1]^n$
using a uniform distribution. Then, with probability 1,
$W_{D}$ will eventually have
$\ww(h)$ as its unique minimal wiring diagram. Equivalently,
with probability 1, $\mathcal{I}_{D}$ will eventually be equal
to $\mathcal{I}_{\ww(h)}$.

\end{theorem}

\begin{proof}
If $h$ is constant, then $\mathcal{I}_{w(h)}=\mathcal{I}_{\{\}}=\langle 0\rangle=\mathcal{I}_D$ for all observed data $D$.\\ 
If $h$ is not constant, without loss of generality we assume that $w(h)=\{(x_1,1),\ldots,(x_k,1)\}$.  Since $h$ is increasing (and not constant) with respect to $x_1$, there exists $p,\bar{p}\in [0,1]^n$ such that $p_i=\bar{p}_i$ for $i\geq 2$ and $p_1<\bar{p}_1$ and $h(p)<h(\bar{p})$. If $p$ or $\bar{p}$ happen to be on the boundary of $[0,1]^n$, using continuity we can pick new $p,\bar{p}$ values that are not on the boundary. Now, consider $S=\{(0,s_2,\ldots,s_n):s_i\in \{-1,1\}\}$ and for any element $s\in S$ and for $\delta>0$ define $p^{s}=\bar{p}+\delta s$. Note that $p^{s}$ is simply $\bar{p}$ after modifying all entries but the first one according to the sign pattern given by $s$. By continuity, we can choose $\delta>0$ such that $h(p)<h(p^s)$ for all $s\in S$. Now, since $sign(p^s_1-p_1)=1$ and $sign(p^s_i-p_i)=s_i$ for $i\geq 2$, we obtain
\begin{align*}
\mathcal{I}_{(p,p^{s})}=\langle (x_1-1)(x-s_2)(x-s_3)\ldots (x_n-s_n) \rangle.
\end{align*}
If we denote $P_1=\{p\}\cup \{p^s: s\in S\}$ and $A_1:=\{ (x_1-1)(x_2-s_2)\ldots(x_n - s_n): s_j\in \{-1,1\}\}$, it follows that $\langle x_1-1 \rangle = \langle A_1\rangle\subseteq \mathcal{I}_{(P_1,h|_{P_1})}$. 
By continuity, we can find open sets $B_p$, $B_{p^s}$ such that $\langle x_1-1 \rangle \subseteq \mathcal{I}_{(P_1,h|_{P_1})}$ as long as one point of each open set is selected. If we sample points in $[0,1]^n$ uniformly, with probability 1 we will eventually sample points in these regions. Thus, with probability 1 we will eventually obtain $\langle x_1-1 \rangle \subseteq \mathcal{I}_D$. \\
The same argument shows that with probability 1 we will eventually obtain $\langle x_j-1 \rangle \subseteq \mathcal{I}_D$ for all $j=1,\ldots,k$ and thus $\langle x_1-1,\ldots, x_k-1 \rangle \subseteq \mathcal{I}_D$. 
The proof now follows from the fact that since $w(h)=\{(x_1,1),\ldots,(x_k,1)\}$, $\mathcal{I}_D\subseteq \langle x_1-1,\ldots, x_k-1 \rangle$
 for any observed data $D$. Indeed, if $p,p'$ satisfy $h(p)<h(p')$, then $p_i<p'_i$ for some $i=1,\ldots,k$. This implies that $x_i-1$ is one of the factors of the generator of $\mathcal{I}_{(p,p')}$, so $\mathcal{I}_{(p,p')}\subseteq \langle x_i-1\rangle \subseteq \langle x_1-1,\ldots, x_k-1 \rangle$. Thus, $\mathcal{I}_D\subseteq \langle x_1-1,\ldots, x_k-1 \rangle$.
The argument above can be modified to also work with functions not necessarily continuous. As long as the points of discontinuity form a set of measure zero. Examples of these functions include piecewise defined functions.
\end{proof}
We remark that the proof also works with an unbounded domain, such as $[0,\infty)^n$. Also, any distribution can be used as long as any open subset of the domain can be sampled with positive probability. 
Examples of distributions satisfying this are gamma, beta, log normal, truncated normal, etc.

\subsection{Example}
We consider the discrete dynamical system $f:[0,1]^5\rightarrow [0,1]^5$ given by the equations.
\begin{align*}
f_1 & = \frac{x_1}{1+x_1}\frac{1}{1+x_2^2}\\
f_2 & = \frac{1}{1+x_1x_2}\frac{1}{1+x_5}\\
f_3 & = \frac{x_1^2}{1+x_1^2}\frac{1}{1+x_2}\\
f_4 & = \frac{1}{1+x_2}\\
f_5 & = \frac{x_1}{1+x_1}\frac{x_2}{1+x_2}
\end{align*}

We sampled 30 points in $[0,1]^2$ uniformly at random and obtained Table \ref{table:eg5vars}. We then used Proposition \ref{prop:minimal} to compute the minimal wiring diagrams for each coordinate function of $f$. We did this using the data $D=(P,f|_P)$, for $|P|=20,22,24,26,28,30$ (starting at the top row of Table \ref{table:eg5vars}). The predicted wiring diagrams are shown in Figure \ref{fig:example5vars} where we can see that eventually the predicted wiring diagram coincides with the true wiring diagram.

\begin{table}[h]
\caption{Example of a data set.}
\begin{center}
\begin{tabular}{|c|lllll||lllll|}
\hline
 $P$ &
& \multicolumn{3}{c}{$x$} & & & \multicolumn{3}{c}{$f(x)$}   & \\ \hline
1& .75 & .30 & .17 & .90 & .70 & .39 & .48 & .28 & .77 & .10\\ \hline 
2& .69 & .98 & .71 & .20 & .31 & .21 & .46 & .16 & .51 & .20\\ \hline 
3&.99 & .50 & .31 & .98 & .97 & .40 & .34 & .33 & .67 & .17\\ \hline 
4&.96 & .75 & .04 & .94 & .25 & .31 & .47 & .27 & .57 & .21\\ \hline 
5&.30 & .16 & .26 & .18 & .66 & .23 & .57 & .07 & .86 & .03\\ \hline 
6&.18 & .53 & .15 & .22 & .28 & .12 & .71 & .02 & .65 & .05\\ \hline 
7&.58 & .05 & .62 & .27 & .88 & .37 & .52 & .24 & .95 & .02\\ \hline 
8&.25 & .65 & .06 & .09 & .75 & .14 & .49 & .04 & .61 & .08\\ \hline 
9&.43 & .10 & .73 & .90 & .61 & .30 & .60 & .14 & .91 & .03\\ \hline 
10&.08 & .79 & .79 & .89 & .52 & .05 & .62 & 0.0 & .56 & .03\\ \hline 
11&.17 & .99 & .04 & .28 & .73 & .07 & .49 & .01 & .50 & .07\\ \hline 
12&.70 & .54 & .52 & .63 & .62 & .32 & .45 & .21 & .65 & .14\\ \hline 
13&.57 & .20 & .75 & .22 & .05 & .35 & .85 & .20 & .83 & .06\\ \hline 
14&.73 & .51 & .25 & .48 & .93 & .33 & .38 & .23 & .66 & .14\\ \hline 
15&.20 & .10 & .77 & .05 & .61 & .17 & .61 & .03 & .91 & .02\\ \hline 
16&.65 & .79 & .40 & .85 & .48 & .24 & .45 & .17 & .56 & .17\\ \hline 
17&.34 & .99 & .50 & .58 & .64 & .13 & .46 & .05 & .50 & .13\\ \hline 
18&.92 & .64 & .65 & .71 & .39 & .34 & .45 & .28 & .61 & .19\\ \hline 
19&.53 & .53 & .43 & .54 & .79 & .27 & .44 & .14 & .65 & .12\\ \hline 
20&.47 & .33 & .78 & .58 & .07 & .29 & .81 & .14 & .75 & .08\\ \hline 
21&.14 & .42 & .61 & .65 & .96 & .10 & .48 & .01 & .70 & .04\\ \hline 
22&.49 & .32 & .66 & .48 & .74 & .30 & .50 & .15 & .76 & .08\\ \hline 
23&.78 & .83 & .18 & .50 & .66 & .26 & .37 & .21 & .55 & .20\\ \hline 
24&.85 & .55 & 1.0 & .97 & .93 & .35 & .35 & .27 & .65 & .16\\ \hline 
25&.36 & .84 & .78 & .43 & .66 & .16 & .46 & .06 & .54 & .12\\ \hline 
26&.09 & .23 & .24 & .27 & .39 & .08 & .70 & .01 & .81 & .02\\ \hline 
27&.95 & .74 & .70 & .12 & .18 & .31 & .50 & .27 & .57 & .21\\ \hline 
28&.67 & .84 & .50 & .06 & .47 & .24 & .44 & .17 & .54 & .18\\ \hline 
29&.28 & .76 & .38 & .14 & .03 & .14 & .80 & .04 & .57 & .09\\ \hline 
30&.57 & .97 & .87 & .28 & .64 & .19 & .39 & .12 & .51 & .18\\ \hline 
\end{tabular} 
\end{center}
\label{table:eg5vars}
\end{table}

\begin{figure}[htbp]
  \centering
  \includegraphics[width=4in]{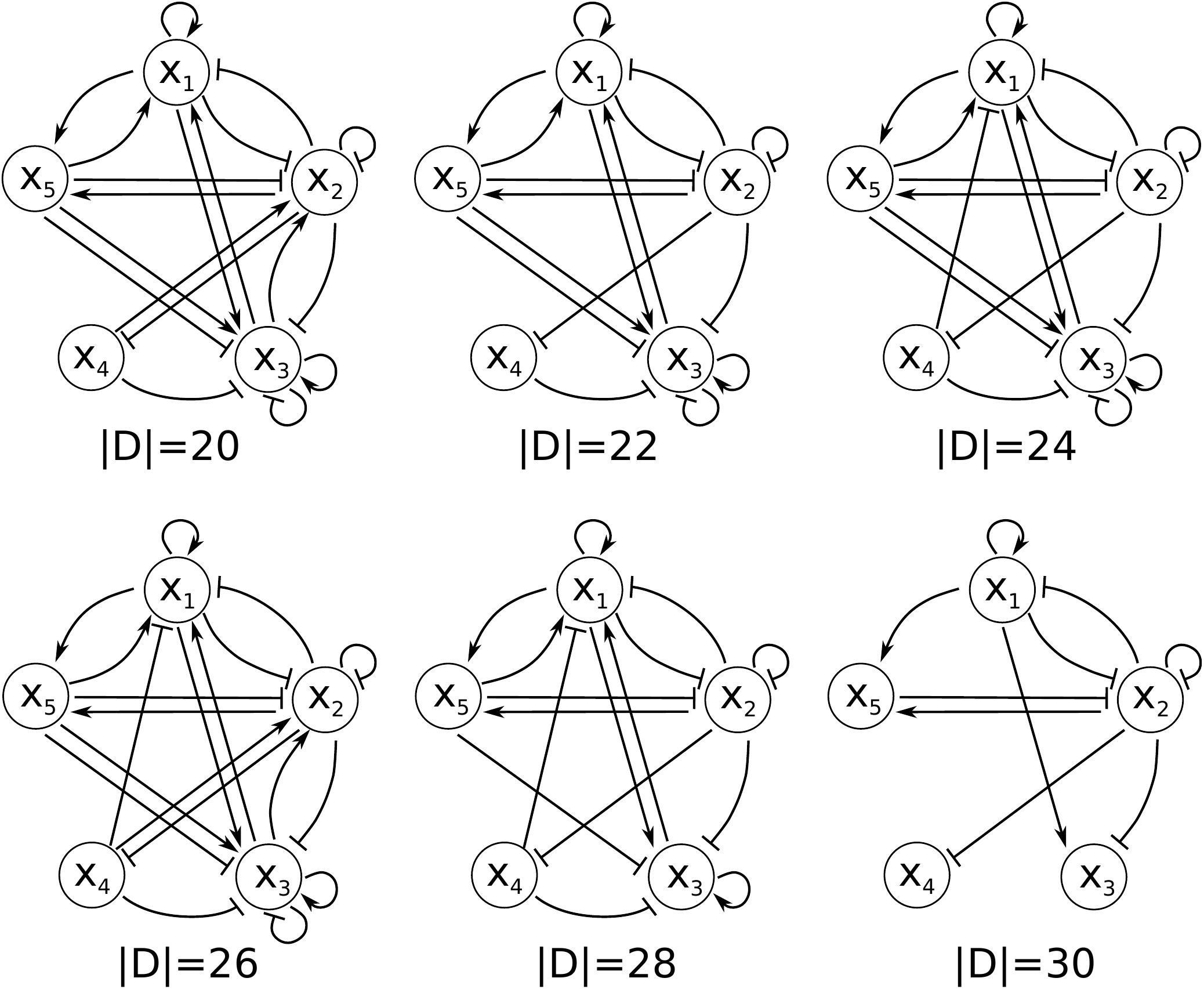}
  \caption{Effect of increasing the number of data points using Table \ref{table:eg5vars}. An edge is drawn if it shows up in at least one minimal wiring diagram. An arrow ($\rightarrow$) indicates a positive sign representing activation and a hammerhead ($\inhib$) indicates a negative sign representing inhibition. The true wiring diagram is reached when $|D|=30$.}
  \label{fig:example5vars}
\end{figure}


\subsection{Flour beetle models}
\begin{example}\rm

We consider the example from page 83 of \cite{Chavez} modeling flour beetle populations.
There are three stages of the beetle life cycle, L (larval), P (pupal), and A (adult).
We assume that the unit of time for the discrete time model describing these populations is 2 weeks.
The linear model is given by the following equations.
\begin{align*}
L_{n+1} & = b A_n \\
P_{n+1} & = (1 - \mu_L) L_n \\
A_{n+1} & = (1 - \mu_P)P_n + (1-\mu_A)A_n
\end{align*}
where the $\mu_i$ is the death rate of stage $i$, and $b$ is the larval recruitment rate per adult in unit time.
Often, $\mu_P$ is taken to be 0.  We instead let it be a small value, e.g. $\mu_P = .003$. The other parameter values
we fix as in the book: $b = 7$, $\mu_L = .2$, $\mu_A = .01$.

We simulate data by choosing 10 sets of random initial conditions in the range $[0,1]^3$. 
For each initial condition,
we simulate the model and save time steps $1, \ldots, 5$, obtaining 50 total observations.  We get 3 ideals,
for $L$, the ideal is $(A-1)$, for $P$, the ideal is $(L-1)$, and for $A$, the ideal is $(A-1, P-1)$.  We repeated this computation  1000 times and always obtained these same ideals, which is exactly what one expects.
\end{example}

\begin{example}\rm
Now we consider the extension of the model to include cannibalism, introduced via nonlinear terms, which is given by the following equations:
\begin{align*}
L_{n+1} & = b A_n e^{-c_{EA}A_n} e^{-c_{EL}L_n} \\
P_{n+1} & = (1 - \mu_L) L_n \\
A_{n+1} & = P_n e^{-c_{PA}A_n} + (1-\mu_A)A_n
\end{align*}
Note that the equation $L_{n+1}$ is not monotone in A.  It is not clear what the correct wiring diagram for the governing equations should be since $A$ could be activating or inhibiting $L$, depending on the value of the data. As expected for a non-monotone equation, we do not recover a consistent wiring diagram for this variable. 
Due to the non-monotonicity, different values of the data will result in different wiring diagrams. We remark that for continuous dynamical systems, such paradoxical results have been observed for reconstructing wiring diagrams with different total concentrations (i.e., perturbations to the initial conditions) and relates to biological retroactivity \cite{paradox}.

\end{example}

\section{Network Reconstruction With Noise}
\label{sec:main}

Now we consider the case of imperfect data. We consider $h:[0,1]^n\rightarrow [0,1]$ and assume that due to stochasticity we sample points $(p , h(p)+\eta)$, where $\eta$ denotes an  unknown noise value. We consider the case of bounded noise with bound $\epsilon_{out}$. That is, $|\eta_i|\leq \epsilon_{out} $. Also, we assume $\eta$ satisfies the condition $\rm{Prob}(\eta \in [-r,r])>0$ for all $r>0$. Examples of distributions that satisfy this condition include uniform, gamma, beta, normal, log normal. We denote with $\tilde{h}$ the function that includes the noise.

\begin{definition}
Suppose we have noisy data $D$ such that the noise bound is given by $\epsilon_{out}$. Define the ideal $\mathcal{I}_{\epsilon_{out},D}=
\sum_{\substack{(p,p')\in P^2 \\ \tilde{h}(p')-\tilde{h}(p)>2\epsilon_{out}} } \mathcal{I}_{p,p'}$.
\end{definition}
 
\begin{theorem}\label{thm:noise-prob1}
Consider a monotone function $h:[0,1]^n\rightarrow [0,1]$ and suppose we obtain noisy data, $D$ by sampling points in $[0,1]^n$
using a uniform distribution. If $\epsilon_{out}$ is low enough, 
$W_{D}$ will eventually have
$\ww(h)$ as its unique minimal wiring diagram with probability 1. Equivalently,
with probability 1, $\mathcal{I}_{\epsilon_{out},D}$ will eventually be equal
to $\mathcal{I}_{\ww(h)}$.
\end{theorem}
\begin{proof}
The proof follows the same argument as the proof of Theorem \ref{thm:prob1}. We only need that with probability 1 we will still sample points such that $\tilde{h}(p)+2\epsilon_{out} <\tilde{h}(p')$ (note this implies $h(p) <h(p')$), which  will be the case if $\epsilon_{out}$ is small enough, such as when $\epsilon_{out} \leq \frac{1}{2} \min\{M_i: h \text{ depends on } x_i\}$, where $M_i:=\max \{ |h(p)-h(p')|: p_{1,\ldots,i-1,i+1,n}=p'_{1,\ldots,i-1,i+1,n}\text{ and } p_i\neq p'_i\}$.
\end{proof}

We demonstrate the results of Theorem \ref{thm:noise-prob1} with an example from fish population dynamics.

\begin{example} \rm
We consider the example from \cite{Townley2012,Eager2016} modelling dynamics of fish populations. There are 5 five stages of the fish cycle, which we denote by $A, B, C, D, E$. The fish cycle stages are ordered $A$ to $E$, corresponding to the initial and final stages. The nonlinear model is given by the following equations.
\begin{align*}
A_{n+1} & = \varphi(k_C C_n + k_D D_n + k_E E_n) \\
B_{n+1} & = s_A A_n   \\
C_{n+1} & = s_B B_n  \\
D_{n+1} & = s_C C_n  \\
E_{n+1} & = s_D D_n
\end{align*}
where the $s_X$'s represent the transition rate from one stage to the next, $\varphi(x)=\frac{Vx}{K+x}$ is a function that describes how fecundity of fish depends on density, and  $k_X$'s capture how each stage contributes to fecundity. Since $0\leq A_{n+1}\leq V$ (the maximum value of $\varphi$) and $0\leq s_X\leq 1$, it follows that the states of the discrete dynamical system will be in $[0,V]^5$ after one iteration. Therefore, we can consider the system to be defined on $[0,V]^5$, i.e., $f:[0,V]^5 \to [0,V]^5$.
The parameters used by the authors were $s_A=0.0131$, $s_B=0.8$, $s_C=0.7896$, $s_D=0.6728$, $k_C=2.2834$, $k_D=35.1099$, $k_E=277.6529$, $K=8$, $V=6$ \cite{Townley2012,Eager2016}.

We simulate data by choosing $N$ sets of random initial conditions in the range $[0,6]^5$.  For each initial condition,
we simulate the model by applying the function for one time step, obtaining $N$ total observations.  We consider measurement noise on the observations by assuming points have the form $(p, f(p) + \eta)$ where $\eta$ is an unknown noise vector in the bounded range $|\eta_i|\leq \epsilon_{out}$.
We think of the data $D$ as a list of pairs $(p, f(p)+\eta)$, where the value of $\eta$ ranges logarithmically between $.5 \times 10^{-9}$ to $.5\times 10^{-1}$. For a fixed $\epsilon_{out}$, we compute the ideal $\mathcal{I}_{\epsilon_{out},D}$ of variable $A$ and $B$, as given in Figure~\ref{fig:fish-output-error}. Notice that Figure~\ref{fig:fish-output-error}A requires more observations for perfect reconstruction than Figure~\ref{fig:fish-output-error}B, which is due to the complexity of the function and relative values $k_C \ll k_D \ll k_E$. 
\begin{figure}[htbp]
  \centering
  \includegraphics[width=\textwidth]{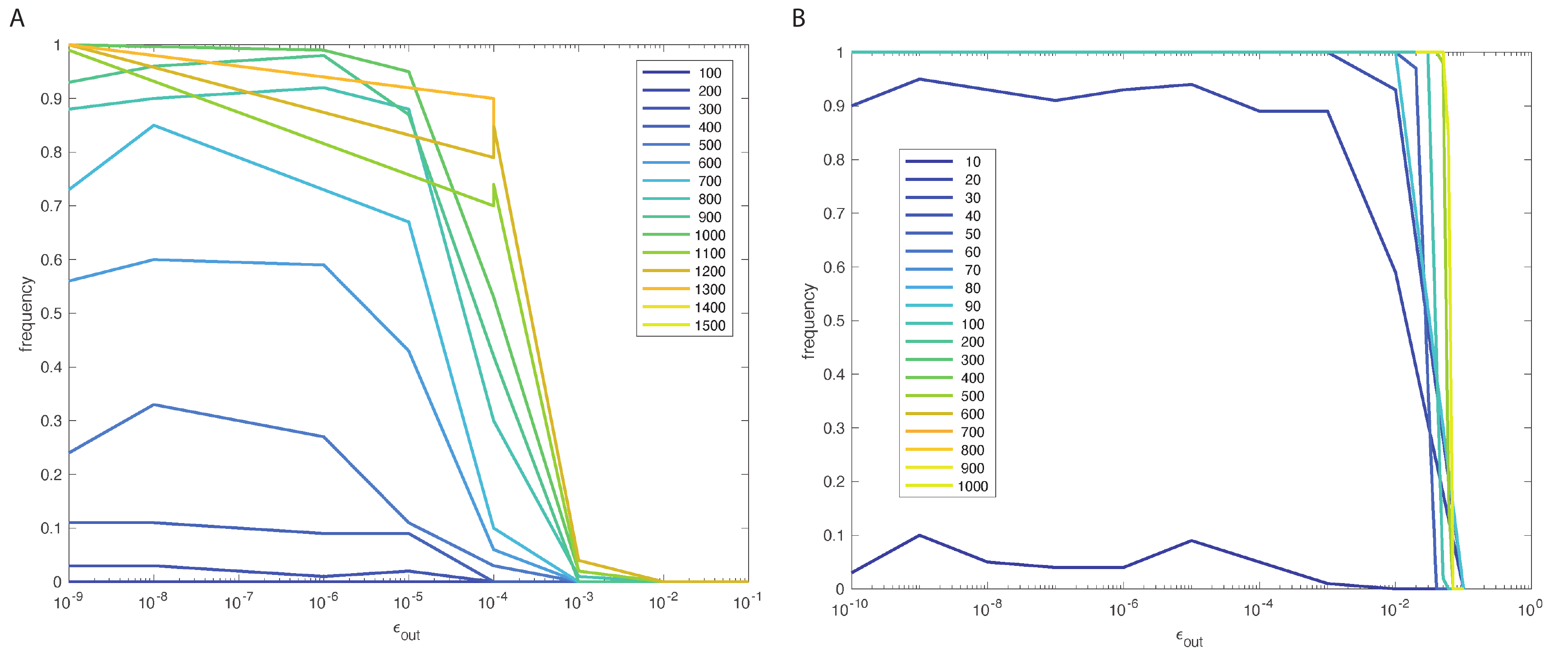}
  \caption{Frequency of correct wiring diagram from Example 4.3. $\epsilon_{out}$ varies logarithmically. (A) Frequency of correct local wiring diagram of variable A with number of observations ranging from 100 (dark blue) to 1500 (yellow). (B) Frequency of correct local wiring diagram of variable B with number of observations ranging from 10 (dark blue) to 1000 (yellow).}
  \label{fig:fish-output-error}
\end{figure}

\end{example}

Now we consider the case of imperfect data in the input and output. We still consider $h:[0,1]^n\rightarrow [0,1]$, but assume that instead of sampling points of the form $(p,h(p))$, we sample points $(p+\xi , h(p)+\eta)$, where $\xi$ and $\eta$ denote  unknown noise vector and value, respectively. The noise in $p$ can be considered measurement error, and the noise in $h(p)$ can be considered measurement error and stochastic noise. We consider the case of bounded noise with bounds $\epsilon_{in}$ for the input and $\epsilon_{out}$ for the output. That is, $|\xi_i|\leq \epsilon_{in}$ and $|\eta_i|\leq \epsilon_{out}$. We denote $\epsilon:=(\epsilon_{in},\epsilon_{out})$. We use $\tilde{}$ to denote quantities with a noise realization. So, we denote with $\tilde{h}$ the function that includes the noise and will denote with $\tilde{p}$ the vector $p+\xi$. Then, the data observed is denoted by $(\tilde{P},\tilde{h}|_P)$. Note that $P$ does not have noise in the second entry since the noise in $P$ is measurement noise.

We are interested in how the wiring diagram changes for observed data with limited precision.
We consider the imperfect observed data $D=(\tilde{P},\tilde{h}|_P)$. 
We now provide definitions that include potential measurement noise in the observed data $D$.  

\begin{definition}\rm
Fix $\epsilon = (\epsilon_{in}, \epsilon_{out})$, where $\epsilon_{in}, \epsilon_{out} \geq 0$. 
Let $(\tilde{p},\tilde{p}')\in P^2$ such that $\tilde{h}(p')-\tilde{h}(p)> 2\epsilon_{out}$. Define the ideal
\[
\mathcal{I}_{\epsilon,(\tilde{p},\tilde{p}')} := \left\langle \prod_{|\tilde{p}_i' -\tilde{p}_i|\ge 2\epsilon_{in},\, \tilde{p}_i \ne \tilde{p}'_i} (x_i-\sgn(\tilde{p}'_i-\tilde{p}_i))\prod_{|\tilde{p}_i' -\tilde{p}_i| < 2\epsilon_{in}}(x_i^2 - 1) \right\rangle.
\]
This ideal is generated by a single polynomial which is a product of linear polynomials.
\end{definition}

The intuition behind including the polynomials $x_i^2-1$ is that if $|p_i'-p_i|<2\epsilon_{in}$, we cannot be sure if the $i$-th input increased or decreased (since $p_i'$ and $p_i$ include a noise term).


\begin{definition}\rm
  Fix $\epsilon = (\epsilon_{in}, \epsilon_{out})$, where $\epsilon_{in}, \epsilon_{out} \geq 0$.   
  Let $D = (\tilde{P}, \tilde{h}|_P)$ be observed data.  Define the ideal
  \[\mathcal{I}_{\epsilon,D}:=\sum_{\substack{(\tilde{p},\tilde{p}')\in P^2 \\ \tilde{h}(p')-\tilde{h}(p)>2\epsilon_{out}}}  \mathcal{I}_{\epsilon,(\tilde{p},\tilde{p}')}.
  \]

  This is an ideal generated by a number of nonlinear polynomials, each is a product of
  linear polynomials.
\end{definition}


Note that if $\epsilon = (0.0, 0.0)$, then these ideals are the ones defined in Section~\ref{sec:alg}.
The following lemma describes how these ideals change as one varies the $\epsilon$.

\begin{lemma}
  (1) If $\epsilon_{in} = \epsilon_{in}'$, and $\epsilon_{out}' \le \epsilon_{out}$, then
  \[\mathcal{I}_{\epsilon_{in}, \epsilon_{out}', D} \supseteq \mathcal{I}_{\epsilon_{in}, \epsilon_{out}, D}.\]
  
  (2) If $\epsilon_{out} = \epsilon_{out}'$, and $\epsilon_{in}' \le \epsilon_{in}$, then
  \[ \mathcal{I}_{\epsilon_{in}', \epsilon_{out}, D} \supseteq \mathcal{I}_{\epsilon_{in}, \epsilon_{out}, D}. \]

  (3) Let $\epsilon = (\epsilon_{in}, \epsilon_{out})$ and $\epsilon' = (\epsilon_{in}', \epsilon_{out}')$.
  If $\epsilon_{out}' \ge \epsilon_{out}$ and $\epsilon_{in} \ge \epsilon_{in}'$ then
  \[ \mathcal{I}_{\epsilon',D} \supseteq \mathcal{I}_{\epsilon,D} \]
\end{lemma}

\begin{proof}
  Notice that each summand $\mathcal{I}_{\epsilon, (p,p')}$ depends
  only on the value $\epsilon_{in}$, and as $\epsilon_{in}$
  increases, then the terms concerned can only change to $x_i^2-1$, and therefore the ideal $\mathcal{I}_{\epsilon, (p,p')}$ can
  only get smaller.

  Therefore, if $\epsilon_{out} = \epsilon_{out}'$, this proves (b).

  For (a), note that for $\epsilon_{in}$ constant, the summands $\mathcal{I}_{\epsilon, (p,p')}$
  do not change, but as $\epsilon_{out}$ increases, the number of these summands can only decrease, and therefore
  the ideal can only get smaller.

  For (c), combine (a) and (b).
\end{proof}

\begin{theorem}\label{thm:noise_inout-prob1}
Consider a monotone continuous function $h:[0,1]^n\rightarrow [0,1]$ and suppose we obtain noisy data, $D$ by sampling points in $[0,1]^n$
using a uniform distribution. If $\epsilon_{in}$ and $\epsilon_{out}$ are small enough, 
$W_{D}$ will eventually have
$\ww(h)$ as its unique minimal wiring diagram with probability 1. Equivalently,
with probability 1, $\mathcal{I}_{\epsilon,D}$ will eventually be equal
to $\mathcal{I}_{\ww(h)}$.
\end{theorem}
\begin{proof}
The proof is very similar to the proof of Theorem \ref{thm:prob1}, but there are some subtle differences such as the need of the factors $(x_i-1)(x_i+1)$ in the definition of the ideals $\mathcal{I}_{\epsilon,D}$. For completeness, we include all the details of the proof.

If $h$ is constant, then $\mathcal{I}_{w(h)}=\mathcal{I}_{\{\}}=\langle 0\rangle=\mathcal{I}_{\epsilon,D}$ for all observed data $D$.

If $h$ is not constant, without loss of generality we assume that $w(h)=\{(x_1,1),\ldots,(x_k,1)\}$.  Since $h$ is increasing (and not constant) with respect to $x_1$, there exists $p,\bar{p}\in [0,1]^n$ such that $p_i=\bar{p}_i$ for $i\geq 2$ and $p_1<\bar{p}_1$ and $h(p)<h(\bar{p})$. If $p$ or $\bar{p}$ happen to be on the boundary of $[0,1]^n$, using continuity we can pick new $p,\bar{p}$ values that are not on the boundary. Now, consider $S=\{(0,s_2,\ldots,s_n):s_i\in \{-1,1\}\}$ and for any element $s\in S$ and for $\delta>0$ define $p^{s}=\bar{p}+\delta s$. Note that $p^{s}$ is simply $\bar{p}$ after modifying all entries but the first one according to the sign pattern given by $s$. By continuity, we can choose $\delta>0$ such that $h(p)<h(p^s)$ for all $s\in S$. Now, since $sign(p^s_1-p_1)=1$ and $sign(p^s_i-p_i)=s_i$ for $i\geq 2$, we obtain
\begin{align*}
\mathcal{I}_{(p,p^{s})}=\langle (x_1-1)(x-s_2)(x-s_3)\ldots (x_n-s_n) \rangle.
\end{align*}
Furthermore, by continuity, for $\epsilon_{in}$ and $\epsilon_{out}$ small enough we also obtain that 
$h(p)+\eta + 2\epsilon_{out} <h(p^s)+\eta'$ and $|(p^s_i+\xi'_i)-(p_i+\xi_i)|>2\epsilon_{in}$ for all $s\in S$ and noise terms $\xi$, $\xi'$, $\eta$, $\eta'$. That is, $\tilde{h}(p)+ 2\epsilon_{out} <\tilde{h}(p^s)$ and $|\tilde{p}^s_i-\tilde{p}_i|>2\epsilon_{in}$

Then,
\begin{align*}
\mathcal{I}_{\epsilon, (\tilde{p},\tilde{p}^{s})}=\langle (x_1-1)(x-s_2)(x-s_3)\ldots (x_n-s_n) \rangle.
\end{align*}

If we denote $P_1=\{p\}\cup \{p^s: s\in S\}$, $D_1=(\tilde{P}_1,\tilde{h}|_{P_1})$, and $A_1:=\{ (x_1-1)(x_2-s_2)\ldots(x_n - s_n): s_j\in \{-1,1\}\}$, it follows that $\langle x_1-1 \rangle = \langle A_1\rangle\subseteq \mathcal{I}_{\epsilon,D_1}$ for any observed realization of $D_1$. 
By continuity, we can find open sets $B_p$, $B_{p^s}$ such that $\langle x_1-1 \rangle \subseteq \mathcal{I}_{\epsilon,D_1}$ as long as one point of each open set is selected.  If we sample points in $[0,1]^n$ uniformly, with probability 1 we will eventually sample points in these regions. Thus, with probability 1 we will eventually obtain $\langle x_1-1 \rangle \subseteq \mathcal{I}_{\epsilon,D}$. \\
The same argument shows that with probability 1 we will eventually obtain $\langle x_j-1 \rangle \subseteq \mathcal{I}_{\epsilon,D}$ for all $j=1,\ldots,k$ and thus $\langle x_1-1,\ldots, x_k-1 \rangle \subseteq \mathcal{I}_{\epsilon,D}$. 
The proof now follows from the fact that since $w(h)=\{(x_1,1),\ldots,(x_k,1)\}$, $\mathcal{I}_{\epsilon,D}\subseteq \langle x_1-1,\ldots, x_k-1 \rangle$
 for any observed data $D$. 
Indeed, if points observed $(p+\xi,\tilde{h}(p))$ and $(p'+\xi',\tilde{h}(p'))$ satisfy $\tilde{h}(p')-\tilde{h}(p)>2\epsilon_{out}$, then $h(p')>h(p)$. Since $h$ is increasing on variables $x_1,\ldots, x_k$, then $p'_i>p_i$ for some $i=1,\ldots,k$. If $|(p'_i+\xi'_i)-(p_i+\xi_i)|\ge 2\epsilon_{in}$, then $x_i-1$ is a factor of the  generator of $\mathcal{I}_{\epsilon,(p+\xi,p'+\xi')}$. On the other hand, if  $|(p'_i+\xi'_i)-(p_i+\xi_i)|< 2\epsilon_{in}$, then $(x_i-1)(x_i+1)$ is a factor of the  generator of $\mathcal{I}_{\epsilon,(p+\xi,p'+\xi')}$. In either case, $\mathcal{I}_{\epsilon,(p+\xi,p'+\xi')} \subseteq \langle x_i-1\rangle \subseteq \langle x_1-1,\ldots, x_k-1 \rangle$. Thus, $\mathcal{I}_{\epsilon,D}\subseteq \langle x_1-1,\ldots, x_k-1 \rangle$.
\end{proof}

We say that a wiring diagram $w$ is consistent with noisy data $D=(\tilde{P},\tilde{h}(P))$ if there exists a function $h^*$ with wiring diagram $w^*\subseteq w$ such that $h^*(P)=h(P)$. The next theorem states that all consistent wiring diagrams are encoded by $ \mathcal{I}_{\epsilon,D}$ even in the presence of noise.

\begin{theorem}\label{thm:noise_inout_exists}
Suppose we obtain noisy data $D$ by sampling points in $[0,1]^n$ using a uniform distribution. For $\epsilon_{in}$ and $\epsilon_{out}$ small enough,  with probability 1 a wiring diagram $w$ is consistent with $D$ if and only if $ \mathcal{I}_{\epsilon,D}\subseteq \mathcal{I}_w $.

\end{theorem}

\begin{proof}
Suppose that $w$ is consistent with $D$ and denote with $h^*$ the corresponding function with wiring diagram $w^*$. Since $\mathcal{I}_{w^*}\subseteq \mathcal{I}_w$, it is enough to prove that $ \mathcal{I}_{\epsilon,D} \subseteq \mathcal{I}_{w^*} $. We will show that any generator $G$ of $ \mathcal{I}_{\epsilon,D} $ is in $\mathcal{I}_{w^*} $. 

Consider a generator of $\mathcal{I}_{\epsilon,D}$: 
\[ G = \prod_{|\tilde{p}_i' -\tilde{p}_i|\ge 2\epsilon_{in},\, \tilde{p}_i \ne \tilde{p}'_i} (x_i-\sgn(\tilde{p}'_i-\tilde{p}_i))\prod_{|\tilde{p}_i' -\tilde{p}_i| < 2\epsilon_{in}}(x_i^2 - 1) \]
where $\tilde{h}(p')-\tilde{h}(p) > 2\epsilon_{out}$. Note that this implies ${h}(p')>{h}(p)$ and hence $h^*(p')>h^*(p)$.

Since $h^*(p')>h^*(p)$ then $p'_i>p_i$ for some $i$ such that $x_i$ is an activator of $h^*$ or $p'_i<p_i$ for some $i$ such that $x_i$ is a repressor of $h^*$. Note that in this case $x_i-\sgn(\tilde{p}'_i-\tilde{p}_i)$ is one of the generators of $\mathcal{I}_{w^*}$.

Now we have two cases: 
If $|\tilde{p}'_i-\tilde{p}_i|\ge 2\epsilon_{in}$, then $x_i-\sgn(\tilde{p}'_i-\tilde{p}_i)=x_i-\sgn(p'_i-p_i)$ is a factor of $G$.
If $|\tilde{p}'_i-\tilde{p}_i|< 2\epsilon_{in}$, then $(x_i^2-1)$ is a factor of $G$, which implies that $x_i-\sgn(p'_i-p_i)$ is a factor of $G$. In any case, a factor of $G$ is one of the generators of $\mathcal{I}_{w^*}$ and hence $G$ is in $\mathcal{I}_{w^*}$.  
This proves that $\mathcal{I}_{\epsilon,D}\subseteq \mathcal{I}_w $.

Now, suppose that $\mathcal{I}_{\epsilon,D}\subseteq \mathcal{I}_w$ and without loss of generality assume $w=\{(x_1,1),\ldots,(x_k,1)\}$. 
Then, define $Q=\{(p_1,\ldots,p_k):p\in P\}$.
  
First, we claim that the data $D'=(Q,h|_P)$  is monotone increasing. That is, for $q,q'\in Q$ and  $v=h(p)$ and $v'=h(p')$,  if $q\leq q'$ (entrywise) then $v\leq v'$ (note that we are not saying $h$ is monotone). By contradiction suppose $h(p')=v'<v=h(p)$ and consider $\epsilon_{out}$ such that $\tilde{h}(p)-\tilde{h}(p')>2\epsilon_{out}$. Since
$\mathcal{I}_{\epsilon,(\tilde{p}',\tilde{p})}\subseteq\mathcal{I}_{\epsilon,D}\subseteq\mathcal{I}_w$, 
there is $i$ such that $(x_i,1)\in w$ and $x_i-1$ is a factor of the generator of $\mathcal{I}_{\epsilon,(\tilde{p}',\tilde{p})}$. This can only happen if $q_i-q'_i\geq 2\epsilon_{in}$ (so $q_i>q'_i$, a contradiction) or if $|q_i-q'_i|<2\epsilon_{in}$. Since $q\leq q'$ (entrywise), choosing $\epsilon_{in}$ small will mean that $q_i=q'_i$ which happens with probability zero. Thus, the data are monotone.

Second, we extend the data to cover a rectangular grid of values. Namely, for $y\in [0,1]^k$, we define 
$
g(y):=\max\{ v : y\leq q \text{ and } (q,v)\in D'\}.
$ We remark that $g$ is monotone increasing, so the data we obtain by restricting $g$ to a rectangular grid will also be monotone increasing.

Third, since we have monotone data on a rectangular grid, we can use multilinear interpolation to obtain a continuous function $L:[0,1]^k\rightarrow [0,1]$ that fits the data on a grid. Then, if we define $h^*:[0,1]^n\rightarrow [0,1]$ by $h(x)=L(x_1,\ldots,x_k)$, it follows that $h^*|_P=h|_P$ and $\ww(h^*)\subseteq w$. This completes the proof.
\end{proof}

%
%
%

\section{Selection of Wiring diagrams}
\label{sec:score}

Our results show that for enough data points, we will obtain the true wiring diagram. However, we need a scoring method to compare the minimal wiring diagrams when the number of data points is not large enough. We will define the scoring method for a single coordinate function at a time. 


For concreteness we assume that the in-degree follows a power law distribution \cite{Aldana2003,Barabasi509}. That is, suppose that there is a $1 \le k_0\le n$ such that $\text{Pr}(|W_{true}|=k)= \frac{c}{k^\gamma}$, for all $k \ge k_0$, and zero for $k < k_0$, where $W_{true}$ denotes the true local wiring diagram, $\gamma$ is a parameter, and $c=\frac{1}{\sum_{k=k_0}^n \frac{1}{k^\gamma}}$ is a normalization constant. First, consider the set of all local wiring diagrams consistent with data $D$, $W_D$. Remember that in order to compute $W_D$ it is enough to find its minimal local wiring diagrams ($W$ is an element of $W_D$ if and only if it contains some minimal local wiring diagram). Let $W_1,\ldots,W_l$ be the elements of $W_D$ and define $N_k=|\{W_j:|W_j|=k\}|$, $N^+_{ki}=|\{W_j:|W_j|=k \textrm{ and } (x_i,1)\in W_j \}|$, $N^-_{ki}=|\{W_j:|W_j|=k \textrm{ and } (x_i,-1)\in W_j \}|$, and $N_{ki}=N^+_{ki}+N^-_{ki}$. 
Note that $n \ge k \ge k_0$ if and only if $N_k \ne 0$.

\begin{proposition}\label{prop:score}
Up to a rescaling factor, $\text{Pr}((x_i,\pm1) \in W_{true}) = \sum_{k=1}^n \frac { N^\pm_{ki} }{ k^\gamma N_k }.$
\end{proposition}
\begin{proof}
Since there are $N_k$ local wiring diagrams of size $k$, the probability that a wiring diagram $W_j$ of size $k$ is correct is $\text{Pr}(W_{true}=W_j)=\frac{\text{Pr}(|W_{true}|=k)}{N_k} = \frac{c}{k^\gamma N_k }$. Now, for a variable $x_i$ we obtain 
\begin{align*}
\text{Pr}((x_i,\pm 1)\in W_{true})  
& = \sum_{j=1}^l\text{Pr}((x_i,\pm 1)\in W_{true}|W_{true}=W_j) \ \text{Pr}(W_{true}=W_j)
\\
& = \sum_{k=1}^n\sum_{|W_j|=k}\text{Pr}((x_i,\pm 1)\in W_{true}| W_{true}=W_j) \ \text{Pr}(W_{true}=W_j)
\\
& = \sum_{k=1}^n \sum_{\substack{|W_j|=k \\ (x_i,\pm 1)\in W_j}} \ \text{Pr}(W_{true}=W_j)
\\
& = \sum_{k=1}^n \sum_{\substack{|W_j|=k \\ (x_i,\pm 1)\in W_j}} \ \frac{c}{k^\gamma N_k }
\\
& = c \sum_{k=1}^n\frac{N^\pm_{ki}}{k^\gamma N_k }
\end{align*}

Thus, the scores of $(x_i,\pm 1)$ are defined as

\[S^\pm(i)=c\sum_{k=1}^n\frac{N^\pm_{ki}}{k^\gamma N_k},\]
where $0/0:=0$. 
\end{proof}

We remark that this score is a probability. We define the score of $x_i$ as 
\[S(i)=S^+(i)+S^-(i)=c\sum_{k=1}^n\frac{N_{ki}}{k^\gamma N_k}.\]

If we need to score the local wiring diagrams as well, we can use the score (or probability) 

\[S(W)=\prod_{(x_i,1)\in W}S^+(i) \prod_{(x_i,-1)\in W}S^-(i)  \prod_{\substack{(x_i,1)\notin W\\ \text{and} \\ (x_i,-1)\notin W}}(1-S^+(i)-S^-(i)).\]
%
%

In the case that we have previous knowledge of the local wiring diagram, then we can incorporate such information in the scoring methods. For example, if it is known that $x_r$ is an activator, then we would define $N_k=|\{j:|W_j|=k \text{ and } (x_r,1)\in W_j\}|$, $N^\pm_{ki}=|\{j:|W_j|=k, (x_r,1)\in W_j, \textrm{ and } (x_i,\pm1)\in W_j \}|$.

\begin{example} \rm
We consider the fish population example.  The following describes a specific data set $(D, h|_D)$, and the wiring diagram ideals and scoring on each edge that result.  We generated $D$ by taking 40 data points chosen uniformly at random with each coordinate in the range $[0, 6]$.  We focus here on the local wiring diagrams for the variables $A$ and $B$.  We consider the possible ideals $I_{\epsilon, D}$ and the resulting probability scores on the possible 10 edges coming in to these variables.

Since this is a small data set, we restrict to small $\epsilon$: $0 \le \epsilon_{in} \le .1$, and $0 \le \epsilon_{out} \le .03$.  The maximum output error was chosen because greater than about $.04$, the ideal is the zero ideal, that is, there are no pairs of data points with different values greater than twice this value.  At that point, one cannot make any prediction at all, beyond our probability model.

For variable $B$, the actual local wiring diagram is $(A,1)$.  
As shown in Figure~\ref{fig:fish_varB_all}, the scores on the monomial ideals give the correct wiring diagram for variable $B$. For variable $A$, the local wiring diagram is $(C, 1), (D, 1), (E, 1)$; however, only $D$ and $E$ have probability 1 (see Figure~\ref{fig:fish_varA_all}). We conclude from this initial analysis that more data are required, which is consistent with Figure~\ref{fig:fish-output-error}. We next generate 1000 data points for variable $A$
and restrict to a smaller error: $\epsilon$: $0 \le \epsilon_{in} \le .0002$, and $0 \le \epsilon_{out} \le .0003$, which returns the correct local wiring diagram $(C, 1), (D, 1), (E, 1)$ as shown in Figure~\ref{fig:fish_varA_all_100}. Interestingly, on this larger dataset, over all points, there were only 3 monomial ideals, meaning there are at most 3 different color levels in each of these graphs.
\begin{figure}[htbp]
  \centering
  \includegraphics[width=\textwidth]{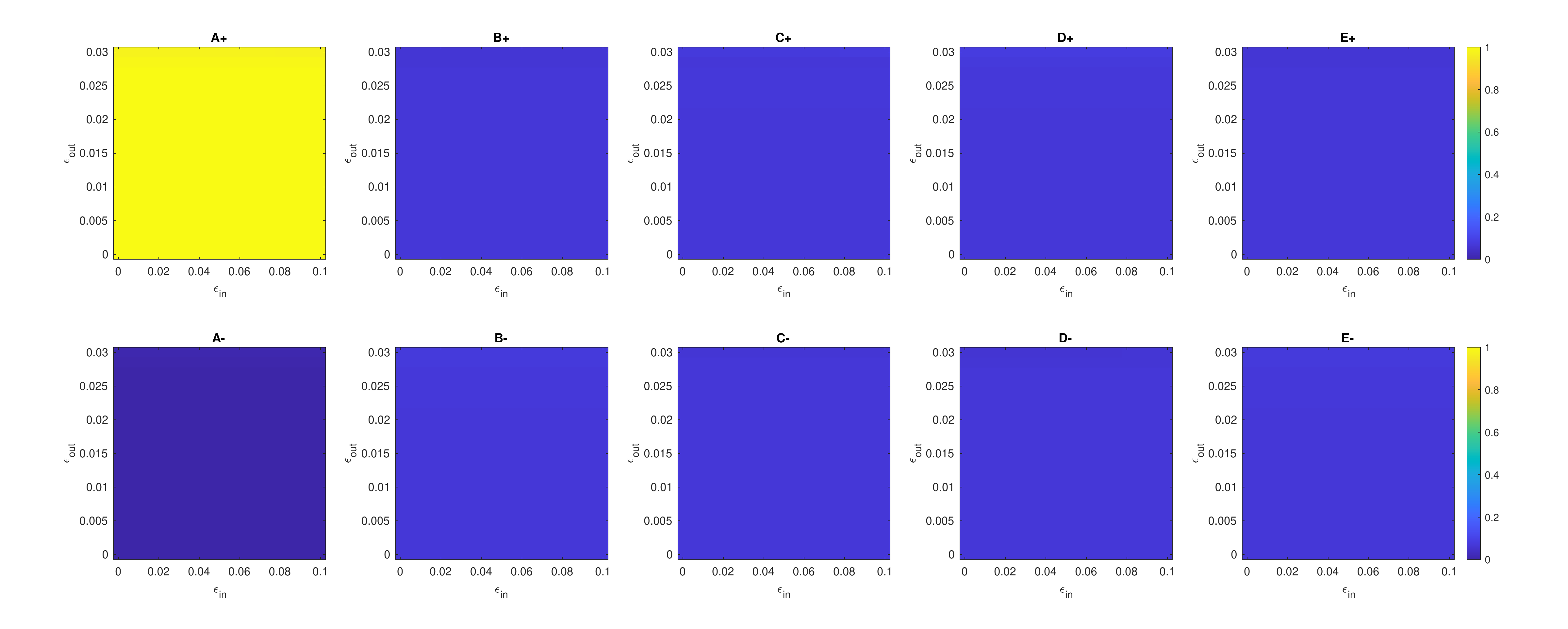}
  \caption{Scores of variable $B$ with 40 data points.}
  \label{fig:fish_varB_all}
\end{figure}

\begin{figure}[htbp]
  \centering
  \includegraphics[width=\textwidth]{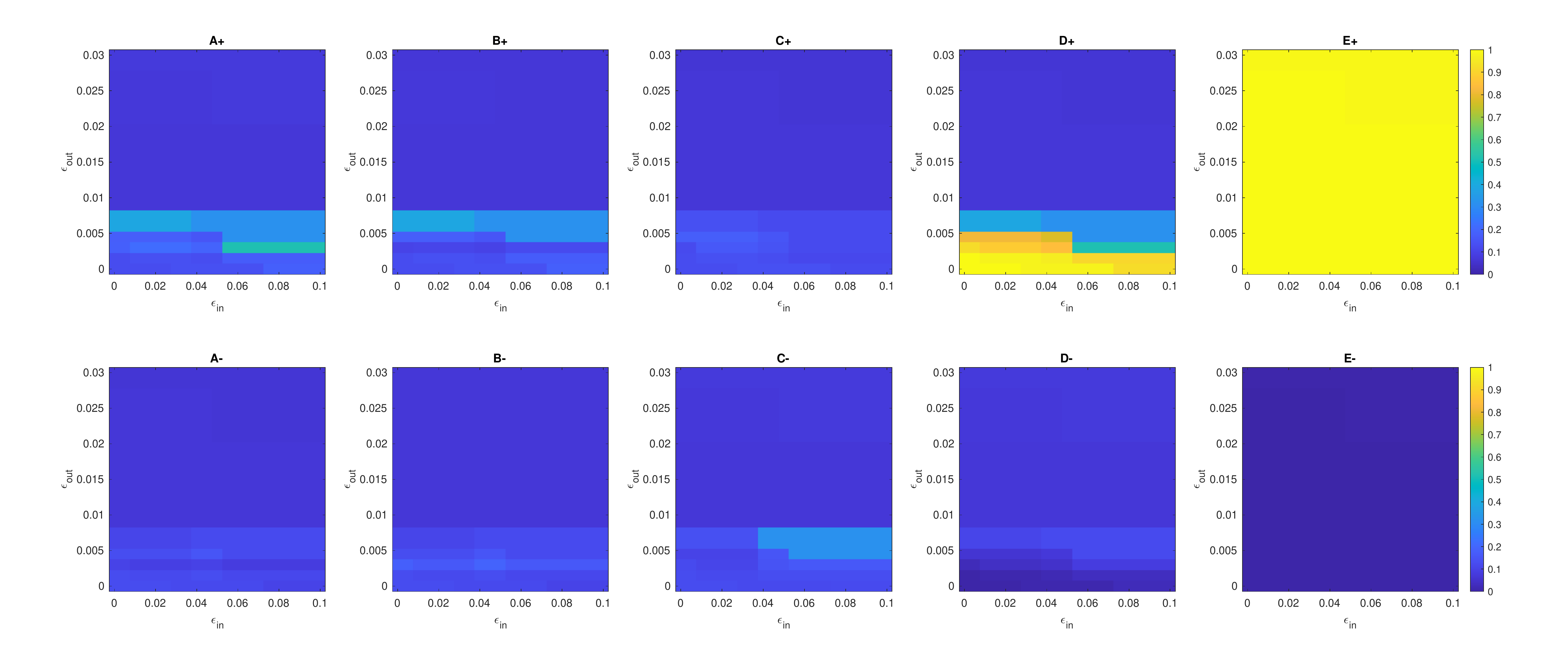}
  \caption{Scores of variable $A$ with 40 data points.}
  \label{fig:fish_varA_all}
\end{figure}

\begin{figure}[htbp]
  \centering
  \includegraphics[width=\textwidth]{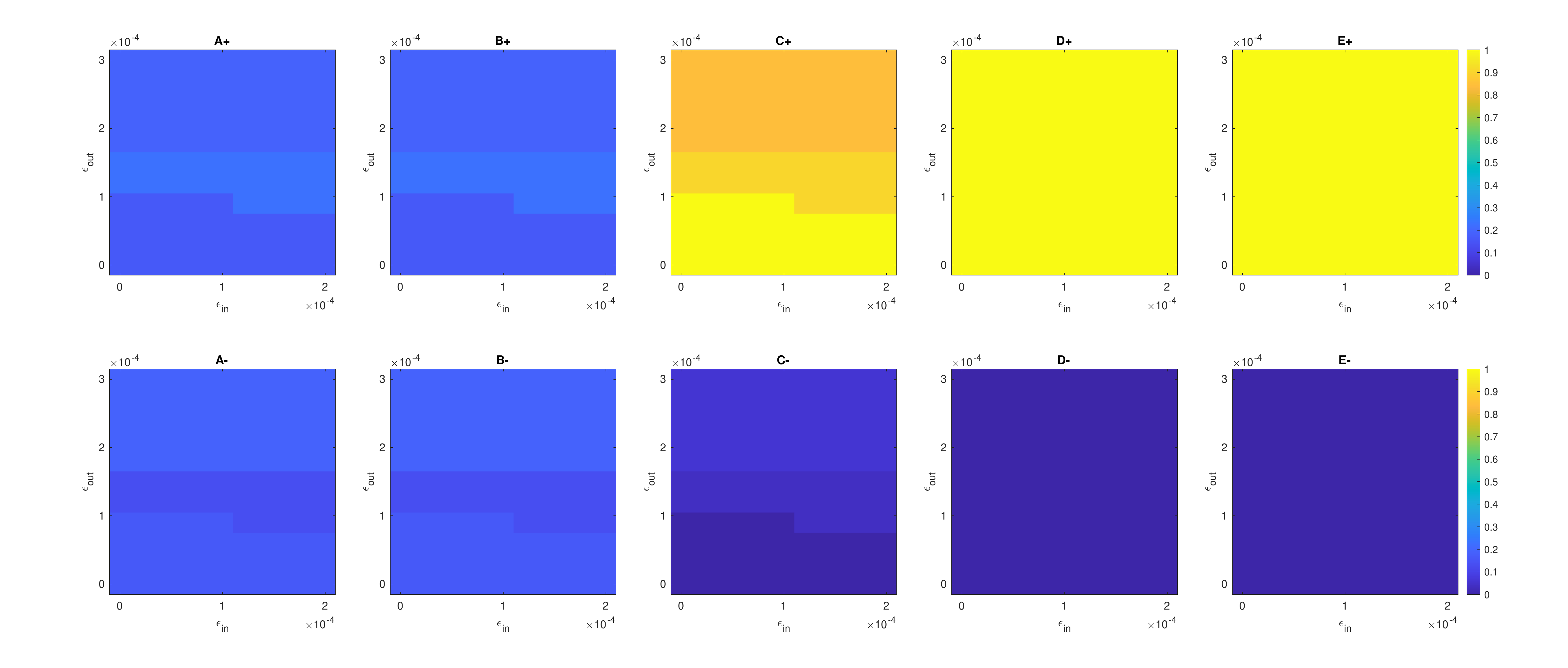}
  \caption{Scores of variable $A$ with 1000 data points.}
  \label{fig:fish_varA_all_100}
\end{figure}

\end{example}

\subsection{Computations}
All computations were performed in Macaulay2 \cite{M2} and will be included as a package in the next M2 distribution. Results from M2 computations were visualised using Matlab.

\section{Conclusions}
\label{sec:conclusions}


We presented an algorithm that computes all minimal wiring diagrams that are consistent with continuous-space data and provides signs to all interactions. Rather than trying to infer model equations and parameter values, our method proposes coarser information at the ``wiring-diagram'' level, without the need to perform parameter estimation. Thus, this method can be used in cases where the functional form of the regulation between variables is possibly unknown. Our algorithm relies on tools from algebraic geometry which has the potential to bring algebraic theory to the problem of reverse engineering. For example, one topic of interest in future is to be able to take measurements in a way that ``maximizes'' information, which algebraically correspond to finding data sets for which we have a unique (or few) irreducible component or prime ideal. Results that show how to minimize the number of irreducible components and prime ideals may provide a theoretical foundation to design experiments in a systematic fashion.


\section*{Acknowledgments}
The first two authors thank Nick Trefethen for suggesting differential monotone spline functions. We thank Hamid Rahkooy for helpful comments on this manuscript. 

\bibliographystyle{siamplain}
\bibliography{references}

\begin{thebibliography}{10}

\bibitem{Aldana2003}
{\sc M.~Aldana}, {\em Boolean dynamics of networks with scale-free topology},
  Physica D: Nonlinear Phenomena, 185 (2003), pp.~45 -- 66,
  \url{https://doi.org/https://doi.org/10.1016/S0167-2789(03)00174-X},
  \url{http://www.sciencedirect.com/science/article/pii/S016727890300174X}.

\bibitem{Arat2015}
{\sc S.~Arat, G.~S. Bullerjahn, and R.~Laubenbacher}, {\em A network biology
  approach to denitrification in pseudomonas aeruginosa}, PLOS ONE, 10 (2015),
  pp.~1--12, \url{https://doi.org/10.1371/journal.pone.0118235},
  \url{https://doi.org/10.1371/journal.pone.0118235}.

\bibitem{Barabasi509}
{\sc A.-L. Barab{\'a}si and R.~Albert}, {\em Emergence of scaling in random
  networks}, Science, 286 (1999), pp.~509--512,
  \url{https://doi.org/10.1126/science.286.5439.509}.

\bibitem{Chavez}
{\sc F.~Brauer and C.~Castillo-Chavez}, {\em Mathematical Models in Population
  Biology and Epidemiology}, Texts in Applied Mathematics, Springer New York,
  2011, \url{https://books.google.com/books?id=nN5RyJf9qMYC}.

\bibitem{dimitrova2022algebraic}
{\sc E.~S. Dimitrova, C.~H. Fredrickson, N.~A. Rondoni, B.~Stigler, and
  A.~Veliz-Cuba}, {\em Algebraic experimental design: Theory and computation},
  arXiv preprint arXiv:2208.02726,  (2022).

\bibitem{Dimitrova2010}
{\sc E.~S. Dimitrova, M.~P.~V. Licona, J.~McGee, and R.~Laubenbacher}, {\em
  Discretization of time series data}, Journal of Computational Biology, 17
  (2010), pp.~853--868, \url{https://doi.org/doi: 10.1089/cmb.2008.0023},
  \url{https://doi.org/10.1089/cmb.2008.0023}.

\bibitem{Eager2016}
{\sc E.~A. Eager and R.~Rebarber}, {\em Sensitivity and elasticity analysis of
  a lur\'e system used to model a population subject to density-dependent
  reproduction}, Mathematical Biosciences, 282 (2016), pp.~34 -- 45,
  \url{https://doi.org/http://dx.doi.org/10.1016/j.mbs.2016.09.016},
  \url{http://www.sciencedirect.com/science/article/pii/S0025556416301997}.

\bibitem{M2}
{\sc D.~R. Grayson and M.~E. Stillman}, {\em Macaulay2, a software system for
  research in algebraic geometry}.
\newblock Available at \url{http://www.math.uiuc.edu/Macaulay2/}.

\bibitem{Jarrah2012}
{\sc A.~Jarrah, R.~Laubenbacher, B.~Stigler, and M.~Stillman}, {\em
  Reverse-engineering of polynomial dynamical systems}, Advances in Applied
  Mathematics, 39 (2007), pp.~477--489.

\bibitem{paradox}
{\sc S.~Prabakaran, J.~Gunawardena, and E.~Sontag}, {\em Paradoxical results in
  perturbation-based signaling network reconstruction}, Biophysical Journal,
  106 (2004), pp.~2720--2728, \url{https://doi.org/10.1103/PhysRevE.68.026121}.

\bibitem{schumaker}
{\sc L.~L. Schumaker}, {\em Spline functions: Computational methods}, Society
  for Industrial and Applied Mathematics, Philadelphia, PA, 2015,
  \url{https://doi.org/10.1137/1.9781611973907.ch1}.

\bibitem{sun2022data}
{\sc J.~Sun, R.~A. Abd~AlRahman, and E.~Bollt}, {\em Data-driven learning of
  boolean networks and functions by optimal causation entropy principle},
  Patterns, 3 (2022), p.~100631.

\bibitem{Townley2012}
{\sc S.~Townley, R.~Rebarber, and B.~Tenhumberg}, {\em Feedback control systems
  analysis of density dependent population dynamics}, Systems \& Control
  Letters, 61 (2012), pp.~309 -- 315,
  \url{https://doi.org/http://dx.doi.org/10.1016/j.sysconle.2011.11.014},
  \url{http://www.sciencedirect.com/science/article/pii/S0167691111002970}.

\bibitem{Veliz-Cuba2012}
{\sc A.~Veliz-Cuba}, {\em An algebraic approach to reverse engineering finite
  dynamical systems arising from biology}, SIAM Journal on Applied Dynamical
  Systems, 11 (2012), pp.~31--48, \url{https://doi.org/10.1137/110828794}.

\bibitem{Velizlacop}
{\sc A.~Veliz-Cuba and B.~Stigler}, {\em Boolean models can explain bistability
  in the \textit{lac} operon}, J. Comput. Biol., 18 (2011), pp.~783--794.

\end{thebibliography}
\end{document}